\documentclass{amsart}

\usepackage{amsmath,amssymb,amsfonts,amsthm}
\usepackage[foot]{amsaddr}

\usepackage{thmtools}
\usepackage{hyperref}
\newtheorem{theorem}{Theorem}
\newtheorem{proposition}[theorem]{Proposition}
\newtheorem{corollary}[theorem]{Corollary}
\newtheorem{claim}[theorem]{Claim}
\newtheorem{lemma}[theorem]{Lemma}

\newtheorem{conjecture}[theorem]{Conjecture}

\begin{document}

\title{Skolem Meets Bateman--Horn}

\author{Florian Luca}
\address{Mathematics Division, Stellenbosch University, Stellenbosch, South Africa}
\email{fluca@sun.ac.za}

\author{James Maynard}
\address{Mathematical Institute, University of Oxford, Oxford, UK}
\email{james.alexander.maynard@gmail.com}

\author{Armand Noubissie}
\address{Institute for Analysis and Number Theory, Graz University of
  Technology, Graz, Austria}
\email{armand.noubissie@tugraz.at}

\author{Jo\"{e}l Ouaknine}
\address{Max Planck Institute for Software Systems\\
Saarland Informatics Campus, Saarbr\"ucken, Germany}
\email{joel@mpi-sws.org}

\author{James Worrell}
\address{Department of Computer Science, University of Oxford\\
  Oxford, UK}
\email{jbw@cs.ox.ac.uk}


\begin{abstract}
The Skolem Problem asks to determine whether a given integer linear
  recurrence sequence has a zero term. This problem arises across a
  wide range of topics in computer science, including loop
  termination, formal languages, automata theory, and control
  theory.
  Decidability is
  notoriously open; the state of the art is a decision procedure for
  recurrences of order at most~4: an advance achieved some 40 years
  ago, based on Baker's theorem on linear forms in logarithms of 
  algebraic numbers.

  A new approach to the Skolem Problem was recently initiated
  in~\cite{LOW,LucaOW22} via the notion of a Universal Skolem Set---a
  set $\mathcal S$ of positive integers such that it is decidable
  whether a given non-degenerate linear recurrence sequence has a zero
  in $\mathcal S$.  Clearly, proving decidability of the Skolem
  Problem is equivalent to showing that $\mathbb{N}$ itself is a
  Universal Skolem Set.  The main contribution of the present paper is
  to construct a Universal Skolem Set
  that has lower density at least $1/8$.
  We show moreover that this set has density $1$ subject to
  Martin's uniform formulation of the Bateman--Horn
  conjecture.  The latter is a far-reaching quantitative
  hypothesis concerning the frequency of primes among the values of
  systems of polynomials.
\end{abstract}

\maketitle

\section{Introduction}
\label{sec:1}
A linear recurrence sequence (LRS) $\boldsymbol u =\langle u_n\rangle_{n=0}^\infty$
is a sequence of integers satisfying a recurrence of the form
\begin{equation}
\label{eq:u}
u_{n+k}=a_0u_{n+k-1}+\cdots+a_{k-1} u_n \qquad(n\in\mathbb{N}) \, ,
\end{equation}
where the coefficients $a_0,\ldots,a_{k-1}$ are integers with
$a_{k-1}\neq 0$.
Equivalently, an integer sequence $\boldsymbol u$ is an LRS if it admits
an exponential-polynomial closed form:
\begin{equation}
u_n = \sum_{i=1}^s Q_i(n) \alpha_i^n \,  ,
\label{eq:ep}
\end{equation}
where $\alpha_1,\ldots,\alpha_s$ are the distinct roots of the
characteristic polynomial of the recurrence~\eqref{eq:u}  and each $Q_i$ is a polynomial
with algebraic coefficients.

The celebrated theorem of Skolem, Mahler, and
Lech~\cite{Lec53,Mah35,Sko34} asserts that the set
$\{n\in\mathbb{N}:u_n=0\}$ of zero terms of an LRS is the union of a
finite set and finitely many arithmetic progressions.  This result can
be refined using the notion of \emph{non-degeneracy} of an LRS\@.  An
LRS is non-degenerate if in its minimal recurrence no quotient of
distinct characteristic roots is a root of unity.  Any given LRS can be
effectively decomposed as the interleaving of finitely many
non-degenerate sequences, some of which may be identically
zero~\cite[Theorem 1.2]{BOOK}.\footnote{More precisely,
the number of subsequences is the least common multiple of periods of
roots of unity arising as quotients of distinct characteristic roots.}
The core of the Skolem-Mahler-Lech
Theorem is the fact that a non-degenerate LRS that is not identically
zero has finitely many zero terms.  Unfortunately, all known proofs
are ineffective---it is not known how to compute the finite set of
zeros of a given non-degenerate linear recurrence sequence;
equivalently, it is not known how to decide whether an arbitrary given
LRS has a zero.

The problem of determining whether a given LRS has a zero is known as
the Skolem Problem.  This has been memorably characterised by
Tao~\cite{Tao08} as \emph{``the halting problem for linear
  automata''}.  In fact, the Skolem Problem has been recognised as a
fundamental decision problem in a number of different areas of
theoretical computer science, including loop
termination~\cite{OuaknineW15}, weighted automata and formal power
series (see~\cite[Section 6.4]{Berstel2010NoncommutativeRS},
\cite[Section 4.2]{RS94}, and~\cite[Section
III.8]{Salomaa1978Automata}), matrix semigroups~\cite{BellPS21},
stochastic systems and probabilistic
programs~\cite{AgrawalAGT15,BJK20}, and control theory~\cite[Section
3]{BlondelT00}.  However, progress towards determining decidability of
the problem has been limited.  The state of the art is that
decidability is known for recurrences of order\footnote{The
  \emph{order} of an LRS is the smallest value of $k$ for which the
  LRS satisfies a recurrence of the form~(\ref{eq:u}).}  at most
4~\cite{MST84,Ver85}---a result established four decades
ago by means of Baker's celebrated theorem on linear forms in logarithms of algebraic
numbers.

Recently~\cite{BiluLNOPW22} gave a procedure to decide the Skolem
Problem for the class of simple LRS (those with simple characteristic
roots) of any order, subject to two classical conjectures about the
exponential function, namely the $p$-adic Schanuel Conjecture and the
exponential local-global principle.   These two
  conjectures are needed to guarantee termination of the procedure,
  but the output (a list of all zeros of a given non-degenerate LRS)
  is unconditionally correct.
The present paper follows a
different approach from~\cite{BiluLNOPW22}, via the notion of
\emph{Universal Skolem Set}.  This is an infinite set
$\mathcal{S} \subseteq \mathbb{N}$ for which there is an effective
procedure that, given a non-degenerate LRS
$\langle u_n\rangle_{n=0}^\infty$, outputs the finite set
$\{ n \in \mathcal{S} : u_n=0\}$.  Evidently, establishing
decidability of the Skolem Problem is equivalent to showing that
$\mathbb{N}$ is a Universal Skolem Set.  Towards this objective, and
noting that the class of Universal Skolem Sets is closed
under various operations such as finite shifts and finite unions, it
is natural to examine diverse means by which to construct Universal
Skolem Sets, and particularly such sets of high density.\footnote{The
  \emph{density} of an infinite set $\mathcal{S}$ of positive integers
  is the limit (if it exists) of the proportion of elements of
  $\mathcal{S}$ among all integers from $1$ to $n$ as $n$ tends to
  infinity. The \emph{lower density} of $\mathcal{S}$ is defined
  analogously, substituting the limit inferior to the limit.}
Pursuing this line of research leads in the present paper to new connections
between the Skolem Problem and classical questions on the distribution
of prime numbers that appear to be unrelated to the $p$-adic
techniques of~\cite{BiluLNOPW22}.

The paper~\cite{LOW} introduced the notion of Universal Skolem Set
(the terminology is inspired by the notion of Universal Hilbert
Set~\cite{Bilu1996}) and exhibited such a set of density zero.
Subsequently~\cite{LucaOW22} produced a set
$\mathcal{S}_0 \subseteq \mathbb{N}$ of positive lower density and an
effective procedure that, given a non-degenerate \emph{simple} LRS
$\langle u_n\rangle_{n=0}^\infty$, computes the set
$\{ n \in \mathcal{S}_0: u_n=0\}$ of zeros of the LRS that lie in
$\mathcal{S}_0$.  The present paper further develops ideas introduced
in~\cite{LucaOW22} and contains two significant advances.  First, we
construct a Universal Skolem Set $\mathcal{S}$ of positive lower
density, i.e., we lift the restriction to simple LRS\@.  
In addition, we give an explicit upper bound for the
largest zero in $\mathcal{S}$ of a given LRS\@.  Such explicit bounds
were lacking in~\cite{LucaOW22}.  The second main contribution is to show that
$\mathcal S$ has lower density at least $1/8$
unconditionally and has density $1$ subject to Martin's
uniform formulation of the Bateman--Horn conjecture, namely Hypothesis
UH~\cite[Hypothesis UH]{Martin02}.  Bateman--Horn is a central unifying
hypothesis concerning the frequency of prime numbers among the values
of a system of polynomials and generalises many classical results and
conjectures, including Hardy and Littlewood's twin primes conjecture
and Schinzel's Hypothesis
H~\cite{BH,ALETHEIAZOMLEFER2020430,Baier02,Lang96}.  The uniform
formulation is quantitative and allows the coefficients of the
polynomials to vary with the main parameter.

Key ingredients of the present paper are deep results of Schlickewei
and Schmidt \cite{SchlickeweiS00} and of Amoroso and Viada~\cite{AV09}
that yield explicit bounds on the number of solutions of certain
exponential-polynomial Diophantine equations.  Indeed, it is striking
that while there is no known method to determine the zero set of a given
non-degenerate LRS, thanks to the above-mentioned results there are
fully explicit upper bounds (depending only on the order of the
recurrence) on the cardinality of its zero set.  Such bounds do not
suffice to solve the Skolem Problem, which would
  require either knowing the precise number of zeros or having
effective upper bounds on the \emph{magnitude} of the zeros of a given
LRS\@.  The main idea of our approach is to leverage explicit upper
bounds on the number of zeros of polynomial-exponential equations to
obtain bounds on the magnitude of the zeros of LRS\@.
We sketch this idea below, eliding many subtle
  technical details.

Roughly speaking,
our Universal Skolem Set $\mathcal S$ consists of positive integers
$n$ that admit sufficiently many \emph{representations} of the form
$n=Pq+a$, with $P,q$ prime and $q,a$ logarithmic in $n$.
More precisely, we ask that $n\in\mathcal S$ have at
  least $\Phi(n)$ representations for some fixed function $\Phi :
  \mathbb N \rightarrow \mathbb N$ that diverges slowly to infinity.
  For the LRS
  $\boldsymbol u$ shown in~\eqref{eq:ep}, $u_n=0$ if
  and only if $\sum_{i=1}^s Q_i(n)\alpha_i^n =0$.  Given a
  representation $n=Pq+a$, substituting  for $n$ into
  the preceding equation and using a generalisation of Fermat's Little
  Theorem, we are able to show that
\begin{equation}
  \sum_{i=1}^s Q_i(a)
  \sigma(\alpha_i)^q\alpha_i^a=0
\label{eq:comp1}
\end{equation}
for some permutation $\sigma$ of the $\alpha_i$.  We
call~\eqref{eq:comp1} a \emph{companion equation}; its form is
determined by the permutation $\sigma$.  Thus, if
$u_n=0$ then each representation of $n$ yields a solution of one of a
finite number of companion equations.  Adapting the above-mentioned
techniques from~\cite{AV09, SchlickeweiS00} we obtain explicit upper
bounds on the number of solutions $q,a \in \mathbb N$ of each
companion equation~\eqref{eq:comp1}.  This in turn yields an 
upper bound on the number of representations of $n$.
Since $n\in \mathcal S$ has at
least $\Phi(n)$ representations and since $\Phi$ diverges to infinity,
we can transfer these upper bounds on the number of representations to
upper bounds on the \emph{magnitude} of~$n$.  We thereby obtain an
explicit upper bound on the magnitude of any zero of the LRS that lies
in the set $\mathcal S$.  In general, we consider that such a transfer
principle---between upper bounds on the multiplicity of solutions of
exponential Diophantine equations and
upper bounds on the magnitude of solutions of such equations---is a promising new
direction to make progress on the Skolem Problem.
Moreover, by choosing $\Phi$ to be sufficiently slowly growing we can
ensure that $\mathcal S$ has positive density.
However, standard
heuristics about the distribution of primes suggest that there are
arbitrarily large integers with no representations at all, which
suggests that our Universal Skolem Set $\mathcal S$ does not contain
all integers.

In terms of proof techniques, a major difference between the present
paper and~\cite{LucaOW22} is that the latter used an existing upper
bound of~\cite{SchlickeweiS00} on the number of solutions of a certain
class of exponential Diophantine equations to bound the number of
solutions of the companion equation.  To handle the case of
non-simple LRS it appears that one cannot use existing results ``off
the shelf'' and must instead adapt the techniques
of~\cite{AV09,ESS,SchlickeweiS00} to our setting.  This is the subject
of Section~\ref{sec:2}, while Section~\ref{sec:density} analyses the density
of the set $\mathcal S$.

\section{Background}
\label{sec:background}
In this section we briefly summarise some notions and results in number theory that we will need.
\subsection{Asymptotic notation}
Given functions $f,g:\mathbb N \rightarrow \mathbb N$, we use the Vinogradov notation
$f \ll g$ to denote that there exists a constant $C>0$ such that $f(n)
\leq C g(n)$ for all but finitely many $n$.  We write $f \sim g$ to
mean that $\lim_{n\rightarrow \infty}\frac{f(n)}{g(n)}=1$.
The same notation is used for real-valued quantities
depending on a real parameter; a subscript, as in $\ll_B$ or
$O_B(\cdot)$, indicates that the implied constant depends on the
subscripted parameter.
For $x>1$ and a positive integer $k\ge 1$, we inductively
define the iterated logarithm function $\log_k x$ as follows:
$\log_1 x:=\log x$, and for $k\ge 2$ we set
$\log_k x:=\max\{1,\log_{k-1} (\log x)\}$.  Thus, for $x$ sufficiently
large, $\log_k x$ is the $k$-fold iterate of $\log$ applied to $x$.  We
omit the subscript when $k=1$.  Dually, $\exp_k$ denotes the $k$-fold
iteration of the exponential function.
\subsection{Number fields}

In this section we recall some facts about number fields.  We refer
to~\cite[Chapter 10]{Baker} for more details. 
Let $\mathbb{K}$ be a finite Galois extension of $\mathbb{Q}$.  We
denote by $\mathrm{Gal}(\mathbb{K}/\mathbb{Q})$ the group of
automorphisms of $\mathbb{K}$.  The ring of algebraic integers in
$\mathbb{K}$ is denoted $\mathcal{O}_{\mathbb{K}}$.  The \emph{norm}
of $\alpha \in \mathbb{K}$ is defined by
\begin{gather*}
N_{\mathbb{K}/\mathbb{Q}}(\alpha) = \prod_{\sigma \in 
  \mathrm{Gal}(\mathbb{K}/\mathbb{Q})} \sigma(\alpha) \, .
\end{gather*}
The 
\emph{norm} $N_{\mathbb{K}/\mathbb{Q}}(\alpha)$ is rational for all 
$\alpha\in\mathbb{K}$ and  $N_{\mathbb{K}/\mathbb{Q}}(\alpha)$ is 
an integer if $\alpha \in \mathcal{O}_{\mathbb{K}}$.
Given $\alpha \in \mathbb K$, $|\alpha|$ denotes its modulus as a
complex number.
  By definition of the norm
we have
$|N_{\mathbb{K}/\mathbb{Q}}(\alpha)| \leq H^{d_{\mathbb{K}}}$, where
$d_{\mathbb{K}}$ is the degree of $\mathbb{K}$ and
\begin{gather*}
  H  := \max_{\sigma \in \mathrm{Gal}(\mathbb{K}/\mathbb{Q})}
  |\sigma(\alpha)|
\end{gather*}
is the \emph{house} of $\alpha$.  Furthermore,
given a rational prime $p \in \mathbb{Z}$ and a prime ideal factor
$\mathfrak{p}$ of $p$ in $\mathcal{O}_{\mathbb{K}}$, we have
$p \mid N_{\mathbb{K}/\mathbb{Q}}(\alpha)$ for all
$\alpha \in \mathfrak{p}$.

We say that $\alpha,\beta \in \mathbb{K}^\times$ are \emph{multiplicatively
  dependent} if there exist integers $k,\ell$, not both zero, such
that $\alpha^k=\beta^\ell$.
\begin{lemma}
\label{lem:multdep}
Let $\sigma \in \mathrm{Gal}(\mathbb{K}/\mathbb{Q})$ and let
$\alpha\in\mathbb{K}^\times$ not be a root of unity.  If
$\alpha^k=\sigma(\alpha)^\ell$ for integers $k,\ell$, then $k=\pm\ell$.
\end{lemma}
\begin{proof}
Repeatedly applying $\sigma$ to the relation
$\alpha^k=\sigma(\alpha)^\ell$ we deduce that
$\alpha^{k^d} = (\sigma^d(\alpha))^{\ell^d}$ for all $d \geq 1$.  In
particular, choosing $d$ to be the order of $\sigma$ we get that
$\alpha^{k^d}=\alpha^{\ell^d}$ and hence $k=\pm \ell$.
\end{proof}
\subsection{Heights}

Let $\overline{\mathbb Q}$ denote the field of algebraic numbers.
For a positive integer $n$, we denote by
$h : \overline{\mathbb{Q}}^n \rightarrow [0,\infty)$ the
\emph{absolute logarithmic Weil height} on $n$-dimensional affine
space.  We refer to~\cite[Chapter 1]{bombierigubler} for the formal
definition.  Here we will need only the following properties
(see~\cite[Chapter 1]{bombierigubler} and~\cite[Corollary 2]{Vou}).

\begin{proposition}
Let $\alpha,\beta,\alpha_1,\ldots,\alpha_s$ be non-zero algebraic numbers
for some $s \geq 2$.   Then
\begin{enumerate}
  \item $h(\alpha) \leq \log H$, where $H$ is the house of $\alpha$;
  \item $h(\alpha)=0$ if $\alpha$ is a root of unity and if $\alpha$
    has degree $d\geq 2$ and is not a root of unity then
  $ h(\alpha) > \frac{2}{d(\log(3d))^3}$;
\item $h(\alpha^m) = |m| h(\alpha)$ for all $m \in \mathbb{Z}$.
\item $h(\alpha\beta) \leq h(\alpha)+h(\beta)$; 
\item $h(\alpha_1+\cdots+\alpha_s) \leq 
  h(\alpha_1)+\cdots+h(\alpha_s)+\log s$;
\item $h(f(\alpha)) \leq (d+1)(h(\alpha)+c+\log 2)$
    if $f(x)\in \overline{\mathbb Q}[x]$ is a
  polynomial of degree~$d$ whose coefficients have height at most $c$;
\item $h(\alpha_1,\ldots,\alpha_s)\geq
  h(\alpha_i/\alpha_j)$ for $i ,j\in \{1,\ldots,s\}$.
\end{enumerate}
\label{thm:height_bounds}
\end{proposition}

The following elementary inequality will be useful in relation to 
calculations involving heights. 
\begin{proposition}
  For all $c,X>1$, if $\sqrt{\log X} < c \log\log X$ then
  one has $X<\exp((4c\log(2c))^2)$.
\label{prop:elementary}
\end{proposition}
\begin{proof}
  Suppose that $X\ge\exp((4c\log(2c))^2)$, so that
  $s:=\sqrt{\log X}\ge 4c\log(2c)>2c$.  The function $s\mapsto
  s-2c\log s$ is increasing for $s>2c$, and its value at
  $s=4c\log(2c)$ equals $2c\log\bigl(c/\log(2c)\bigr)\ge0$, since
  $c\ge\log(2c)$ for all $c>1$.  Hence $s\ge2c\log s$, that is,
  $\sqrt{\log X}\ge c\log\log X$.
\end{proof}

The hypothesis of Proposition~\ref{prop:elementary} will always arise
with an explicitly bounded value of $c$, in which case we will use the
following consequence.

\begin{corollary}
\label{cor:elementary}
Let $k\ge2$ and suppose that $\sqrt{\log X}<c\log_2X$ for some $c>1$
with $\log c\le20k^4$.  Then $X<\exp_2(208k^4)$; in particular,
$2X<\exp_5(10^{10}k^6)$.
\end{corollary}
\begin{proof}
  If $X<e^e$ the first conclusion is trivial; otherwise
  $\log_2X=\log\log X$ and Proposition~\ref{prop:elementary} gives
  $X<\exp((4c\log(2c))^2)$.  From $\log c\le20k^4$ we get
  $c<\exp(20k^4)$ and $\log(2c)<21k^4$, so
  \[
    (4c\log(2c))^2<\bigl(84k^4\exp(20k^4)\bigr)^2<\exp(208k^4),
  \]
  where the last inequality uses $84k^4<\exp(84k^4)$.  Hence
  $X<\exp_2(208k^4)$.  Finally, $208k^4<10^{10}k^6$ and
  $2\exp_2(N)<\exp_5(N)$ for all $N\ge1$, whence
  $2X<\exp_5(10^{10}k^6)$.
\end{proof}

We shall also use the following elementary ``absorption'' lemma, which
allows the various terminal bounds in the proof of
Theorem~\ref{thm:1} to be compared with the bound of that theorem in a
uniform manner.

\begin{lemma}
\label{lem:absorb}
For all integers $k\ge2$ and all real $A\ge10$,
\[
  2\exp_3\bigl((20\,k!\log(kA))^2\bigr)\;\le\;
  \max\{\exp_3(A^2),\,\exp_5(10^{10}k^6)\}.
\]
\end{lemma}
\begin{proof}
  Suppose first that $A>80\,k!\log A$.  Then $A>160\log A$, so
  $A>10^3$; moreover $A>80k$, whence $\log(kA)<2\log A$ and
  \[
    20\,k!\log(kA)<40\,k!\log A<A/2.
  \]
  Consequently the left-hand side is at most
  $2\exp_3(A^2/4)\le\exp_3(A^2)$, using $A^2/4+1\le A^2$ and
  $2\exp_3(t)\le\exp_3(t+1)$ for $t\ge1$.

  Suppose instead that $A\le80\,k!\log A=:B_1\log A$.  Since
  $t\mapsto t/\log t$ is increasing for $t\ge3$ and
  $2B_1\log B_1/\log(2B_1\log B_1)\ge B_1$, it follows that
  $A\le2B_1\log B_1$, and hence
  $\log(kA)\le\log\bigl(160\,k\,k!\log(80\,k!)\bigr)\le3k\log k+8$.
  Using $\log(20\,k!)\le3+k\log k$, we obtain
  \[
    (20\,k!\log(kA))^2
    \le\exp\bigl(6+2k\log k+2\log(3k\log k+8)\bigr)
    \le\exp(4k^2)
    \qquad(k\ge2).
  \]
  Consequently the left-hand side is at most
  $2\exp_4(4k^2)<\exp_5(4k^2)\le\exp_5(10^{10}k^6)$.
\end{proof}

\subsection{Linear equations in elements of multiplicative groups}
\label{sec:lineareq}
We now present two results concerning solutions of linear equations
with variables in multiplicative groups.  These results will play a
key role in our construction of a Universal Skolem Set.  Throughout
this section $\mathbb{K}$ denotes a number field.

The first result  is due to Amoroso and
Viada~\cite[Theorem 6.1]{AV09}.  Here, given $n\geq 1$ and elements
$\boldsymbol{x}=(x_1,\ldots,x_n)$ and
$\boldsymbol{y}=(y_1,\ldots,y_n)$ of $(\mathbb{K}^\times)^n$, we
write $\boldsymbol{x} \ast \boldsymbol{y} := (x_1y_1,\ldots,x_ny_n)$.

\begin{theorem}
  Let $\Gamma$ be a finitely generated subgroup of
  $(\mathbb{K}^\times)^n$
of  rank $r$ and let 
  $\varepsilon:=(8n)^{-6n^3}$.  Then the set
  \[ \left\{ \boldsymbol{x} \ast \boldsymbol{y}  : 
      \boldsymbol{x}\in\Gamma,\,  \boldsymbol{y} \in
      \mathbb{K}^n,\,
      \boldsymbol{x}^\top \boldsymbol{y}=1, \text{ and }
    h(\boldsymbol{y}) \leq \varepsilon(1+ h(\boldsymbol{x})) \right\} \]    
is contained in a union of at most $(8n)^{6n^2(n+r)}$ proper linear
subspaces of $\mathbb{K}^n$.
\label{thm:AV09}
\end{theorem}

The second result, due to Schlickewei and 
Schmidt~\cite{SchlickeweiS00}, concerns equations of the form
\begin{gather}
  \sum_{i=1}^s P_i(\boldsymbol{x}) \boldsymbol{\alpha}_i^{\boldsymbol x}=0
\label{eq:exp-poly}
\end{gather}
in variables $\boldsymbol{x}=(x_1,\ldots,x_n) \in \mathbb{Z}^n$, where
$P_1,\ldots,P_s \in \mathbb{K}[\boldsymbol{x}]$, and
$\boldsymbol{\alpha}_i^{\boldsymbol x} = \alpha_{i1}^{x_1} \cdots
\alpha_{in}^{x_n}$ with $\alpha_{ij} \in \mathbb{K}^\ast$ for all
$i,j$.  We say that a given solution to Equation~\eqref{eq:exp-poly} is
\emph{non-degenerate} if no proper sub-sum
vanishes.  The following 
upper bound on the number of non-degenerate solutions appears
as~\cite[Theorem 1]{SchlickeweiS00}:
\begin{theorem}
  Let $\delta_i \in \mathbb N$ be the degree of polynomial $P_i$ for
  $i\in \{1,\ldots,s\}$.  Put $A_0=\sum_{i=1}^s \binom{n+\delta_i}{n}$
  and $B_0=\max\{n,A_0\}$.  Suppose that there is no non-zero
  $\boldsymbol{x}\in\mathbb{Z}^n$ such that
  $\boldsymbol{\alpha}_i^{\boldsymbol{x}} =
  \boldsymbol{\alpha}_j^{\boldsymbol{x}}$ for all
  $i,j \in \{1,\ldots,s\}$.  Then Equation~\eqref{eq:exp-poly} has at
  most $2^{35B_0^3}d^{6B_0^2}$ non-degenerate solutions, where $d$ is the
  degree of the number field $\mathbb{K}$.
\label{thm:SchSch}
  \end{theorem}

    \subsection{Distribution of primes}
    Define $f=f_1f_2$ for linear forms $f_1(t):=a_1t+b_1$ and $f_2(t)=a_2t+b_2$ for
    integers $a_1,a_2,b_1,b_2$, with $a_1,a_2>0$.
    
We are interested in the frequency of $n \in \mathbb N$ such that $f_1(n)$ and $f_2(n)$
are both prime.  An obstruction to the existence of infinitely many
such $n$ is when
$f(n) \equiv 0 \bmod p$ (for some fixed prime $p$) for all $n\in \mathbb N$.  In this case 
for $f_1$ and $f_2$ to be simultaneously prime, one (or both) of them must equal
$p$, which can happen at most twice.  If $f$ does not vanish identically modulo
$p$ for any prime $p$ then we say that it is \emph{admissible}.
For a prime $p$, let $\omega_f(p)$ denote the number of
$x\in\mathbb{F}_p$ such that $f(x)=0$.

In Section~\ref{sec:count} the coefficients of the two linear forms
vary with $X$, so the usual pointwise Bateman--Horn conjecture is not
sufficient for our application.  We therefore present Martin's
uniform formulation of Bateman--Horn: specifically we
state the degree-two, two-factor case
of Hypothesis UH~\cite[Hypothesis UH, Equation (1.6)]{Martin02}.
For distinct linear forms $f_1,f_2$, put $f:=f_1f_2$,
\[
\pi(f;T):=\#\{1\leq n\leq T: |f_1(n)|\text{ and }|f_2(n)|
\text{ are prime}\}
\]
and
\[
\operatorname{li}(f;T):=
\int_{\substack{0<t<T\\ \min\{|f_1(t)|,|f_2(t)|\}\geq2}}
\frac{dt}{\log|f_1(t)|\log|f_2(t)|} \, .
\]

\begin{conjecture}[Uniform Bateman--Horn (Martin's Hypothesis UH)]
\label{conj:bateman-horn}
  Let $B$ be a sufficiently large fixed positive constant.  Uniformly
for pairs of distinct linear forms $f_1,f_2\in\mathbb Z[t]$ for which
$f=f_1f_2$ is admissible and every coefficient of $f$ has absolute
value at most $T^B$, one has
\begin{gather}
\pi(f;T)=C_f\operatorname{li}(f;T)
+O_B\!\left(\frac{C_fT}{(\log T)^3}+1\right),
\qquad
C_f:=\prod_{p\text{ prime}}\frac{p(p-\omega_f(p))}{(p-1)^2}.
\label{eq:bateman-horn}
\end{gather}
\end{conjecture}
For an admissible product of two distinct linear forms one has
$a_1b_2\ne a_2b_1$ (for a proportional pair, $f$ would have a fixed
prime divisor), so $\omega_f(p)=2$ for every prime
$p\nmid a_1a_2(a_1b_2-a_2b_1)$.  The corresponding local factor is
$1+O(p^{-2})$, and the Euler product defining $C_f$ therefore
converges unconditionally.
For fixed $f$, one has
$\operatorname{li}(f;T)\sim T/(\log T)^2$, so this hypothesis implies
the usual Bateman--Horn asymptotic for a fixed pair of linear forms.
The general Bateman--Horn conjecture concerns simultaneous prime values
of an arbitrary finite family of polynomials; the formulation above is
the uniform degree-two, two-factor case that we need.

Even the fixed-polynomial Bateman--Horn conjecture is known only for a
single linear polynomial, where it is the Prime Number Theorem for
arithmetic progressions.  In particular, the Hardy--Littlewood twin
prime conjecture remains open.

For the unconditional argument we use the standard two-dimensional
Brun--Selberg upper-bound sieve.  In the form needed below, for every
fixed $B>0$, uniformly for pairs of distinct admissible linear forms
$f_1,f_2$ whose coefficients have absolute value at most
$(\log T)^B$, and uniformly for intervals $I\subseteq[0,T]$ with
$|I|\asymp T$, the sieve gives
\begin{gather}
 \#\{t\in I:f_1(t),f_2(t)\text{ are prime}\}
 \leq (\kappa+o_B(1))C_f\frac{|I|}{(\log T)^2},
 \qquad \kappa:=8.
 \label{eq:BRUN}
\end{gather}
This is a special case of a theorem of Halberstam and
Richert~\cite[Theorem~5.7]{HR}: for distinct linear forms
$f_i(t)=a_it+b_i$, $1\le i\le g$, with
$E:=\prod_{i}a_i\prod_{r<s}(a_rb_s-a_sb_r)\ne0$, the number of $t\le x$
such that $f_1(t),\ldots,f_g(t)$ are all prime is at most
\[
  2^g\,g!\,\prod_p\frac{1-\omega_f(p)/p}{(1-1/p)^{g}}
  \cdot\frac{x}{(\log x)^{g}}
  \left(1+O_g\!\left(\frac{\log\log 3x+\log\log 3|E|}{\log x}\right)\right).
\]
For coefficients of absolute value at most $(\log T)^B$ one has
$\log\log|E|\ll_B\log_3T=o(\log T)$, which gives~\eqref{eq:BRUN} for
$I=[0,T]$ with $\kappa=2^2\cdot2!=8$; the case of a general interval
$I\subseteq[0,T]$ with $|I|\asymp T$ follows upon replacing $t$ by
$t+\lfloor\inf I\rfloor$, which changes neither $E$ nor the singular
series $C_f$.

Under the same coefficient and interval hypotheses, suppose in
addition that
\[
 E:=|a_1a_2(a_1b_2-a_2b_1)|\ne0
 \qquad\text{and}\qquad
 \gcd(a_1,a_2)\leq2
\]
(for two forms we take $E$ in absolute value).
Then
\begin{gather}
 \#\{t\in I:f_1(t),f_2(t)\text{ are prime}\}
 \ll_B\frac{E}{\varphi(E)}
       \frac{|I|}{(\log T)^2}.
 \label{eq:sieve}
\end{gather}
Indeed, if the pair is admissible, then for every odd prime
$p\mid E$ the hypothesis on the leading coefficients ensures that
at least one of the forms has a root modulo $p$.  Thus
$\omega_f(p)\geq1$, and the local factor at $p$ is at most
$p/(p-1)$.  At primes $p\nmid2E$ the local factor is at most
$1$, while the factor at $2$ is bounded absolutely.  Consequently
$C_f\ll E/\varphi(E)$, and~\eqref{eq:sieve} follows
from~\eqref{eq:BRUN}.  A non-admissible pair has only $O(1)$
simultaneous prime values, which is absorbed in~\eqref{eq:sieve} for
sufficiently large $T$.

\section{A Universal Skolem Set}
\label{sec:2}
 In this section we define a subset $\mathcal S$ of
  integers and show in Theorem~\ref{thm:1} that it is a Universal
  Skolem Set, by giving an explicit upper bound on
  $\{ n \in \mathcal S : u_n=0\}$ for a given LRS $\boldsymbol u$.  The
  main technical ingredient of the proof of Theorem~\ref{thm:1} are
  the results in Section~\ref{sec:lineareq} that provide upper bounds on
  the number of solutions of exponential-polynomial Diophantine
  equations.

Fix a positive integer parameter $X$.  We define disjoint intervals
\begin{gather}
A(X):=\left[\log_2 X, \sqrt{\log X}\right]
\quad\mbox{and}\quad 
B(X):=\left[\frac{\log X}{\sqrt{\log_3 X}}, \frac{2\log X}{\sqrt{\log_3 X}}\right].
\label{eq:defAB}
\end{gather}
We further define a \emph{representation} of an integer $n\in [X,2X]$
to be a triple $(q,P,a)$ such that $q \in A(X)$, $a \in B(X)$, $P$ and
$q$ are prime, and $n=Pq+a$.
We say that two distinct representations $n=Pq+a$ and
$n=P'q'+a'$ are \emph{correlated} if
\begin{gather*}
  |(a+\eta q)-(a'+\eta q')|<\sqrt{\log X}
\end{gather*}
for some $\eta\in\{\pm1\}$.

We denote by $r(n)$ the number of representations of $n$.  Finally we
put
\begin{align*}
  \mathcal{S}(X) := \left\{ n \in [X, 2X] :
    \begin{aligned}
      & \, r(n) > \log_4 X \text{ and no two}\\
      & \, \text{representations of $n$ are correlated}
    \end{aligned}
  \right\}
\end{align*}
and we define
\begin{gather*}
 {\mathcal S}:=\bigcup_{\substack{X\in\mathbb N\\X\geq2^{10}}}
 \mathcal S(X) \, .
\end{gather*}

Membership in $\mathcal S$ is effective.
The following result shows that $\mathcal S$ is a Universal Skolem Set.
\begin{theorem}
\label{thm:1}
Let ${\boldsymbol u}=\langle u_n\rangle_{n=0}^\infty$ be a non-degenerate LRS of order $k\ge 2$ given by
\begin{gather*}
u_{n+k}=a_0u_{n+k-1}+\cdots+a_{k-1}u_n
\end{gather*}
for $n\ge0$, with given initial terms $u_0,\ldots,u_{k-1}$ not all
zero and $a_{k-1}\ne 0$. 
If $u_n=0$ and $n\in {\mathcal S}$, then
\begin{gather*}
n\leq\max\{\exp_3(A^2),\exp_5(10^{10}k^6)\},\qquad
A:=\max\{10,|u_i|,|a_i|:0\leq i\leq k-1\}.
\end{gather*}
\end{theorem}
The rest of this section is devoted to the proof of
Theorem~\ref{thm:1}, which is divided into 5 steps.
Throughout, we abbreviate the bound of
Theorem~\ref{thm:1} by
\[
  \mathcal T:=\max\{\exp_3(A^2),\,\exp_5(10^{10}k^6)\},
\]
a quantity depending only on $k$ and $A$.
  Our goal is
to show that for all $X$, if $u_n=0$ for some
$n \in \mathcal{S}(X)$ then
\begin{gather}
  2X \leq \mathcal T \,  .
  \label{eq:desired}
\end{gather}
The theorem follows since $n \leq 2X$.

\textbf{Step 1: Rescaling.}
\label{sec:rescaling}
We rescale $\boldsymbol{u}$ so that all the coefficients of the
polynomials in its closed form representation are algebraic integers.
To this end, let 
\begin{gather*}
\Psi(x):=x^k-a_0x^{k-1}-\cdots-a_{k-1}
\end{gather*}
be the characteristic polynomial of ${\boldsymbol u}$ and suppose that
$\Psi$ has roots $\alpha_1,\ldots,\alpha_s$, where $\alpha_i$ has
multiplicity $\nu_i$. Note that all characteristic roots are
algebraic integers since the characteristic polynomial is monic and
has coefficients in $\mathbb{Z}$. The field
${\mathbb K}:={\mathbb Q}(\alpha_1,\ldots,\alpha_s)$, known as the splitting field of $\Psi$,
obtained by adjoining the roots of $\Psi$ to $\mathbb Q$, has
degree at most $k!$.  If 
$|\alpha_i|>1$, then rearranging the equation $\Psi(\alpha_i)=0$ yields
\begin{gather*}
|\alpha_i|=\left|a_0+\frac{a_1}{\alpha_i}+\cdots+\frac{a_{k-1}}{\alpha_i^{k-1}}\right|<kA\, ,
\end{gather*}
for $A$ as in the statement of Theorem~\ref{thm:1}.
Writing $\rho:=\max_{1\leq i \leq s}|\alpha_i|$, we thus have $\rho<kA$.

The sequence $\boldsymbol{u}$ admits a closed-form solution
$ u_n=\sum_{i=1}^s Q_i(n)\alpha_i^n \, , $ where polynomial
$Q_i$ has degree at most $\nu_i -1$ for $i \in \{1,\ldots,s\}$.
The coefficients of the polynomial $Q_i(x)$ belong to $\mathbb{K}$ and 
can be computed from the initial values
$u_0,\ldots,u_{k-1}$ by solving a system of linear equations, thanks to 
the closed-form equation for $u_n$.
By Cramer's
rule, each of the coefficients of $Q_i(x)$ is the quotient of an 
algebraic integer by the determinant\footnote{The underlying matrix
  has $s$ blocks, one for each characteristic root. For $\ell \in \{1,\ldots,s\}$ the
  $\ell$-th block has dimension $k\times \nu_{\ell}$ and has $(i,j)$-th
  element $(i-1)^{j-1} \alpha_{\ell}^{i-1}$
  for $i \in \{1,\ldots,k\}$
and  $j \in \{1,\ldots,\nu_\ell\}$.}
\[
\Delta :=\begin{vmatrix} 1 & \ldots & 0 & 1 & \ldots & 0 & 1 & \cdots\\
\alpha_1 & \ldots & \alpha_1 & \alpha_2 & \ldots & \alpha_{s-1} & \alpha_s & \ldots\\
\vdots & \ddots & \vdots & \vdots & \ddots & \vdots & \vdots & \ddots\\
\alpha_1^{k-1} & \ldots & (k-1)^{\nu_1-1} \alpha_1^{k-1} &
\alpha_2^{k-1} & \ldots & (k-1)^{\nu_{s-1}-1} \alpha_{s-1}^{k-1} &
\alpha_s^{k-1} &\ldots \end{vmatrix} \, .
\]
The length of each column vector above is at most 
\begin{gather*}
{\sqrt{k(k-1)^{2(k-1)} \rho^{2k}}}<k^k (kA)^k=k^{2k} A^k.
\end{gather*}
Thus, by the Hadamard inequality,
$|\Delta|^2<(k^{2k^2} A^{k^2})^2=(k^2 A)^{2k^2}$.
Note also that $\Delta\ne0$: since the characteristic
roots are non-zero, elementary column operations (passing from the
monomials $x^j$ to falling factorials within each block, and scaling)
transform the underlying matrix into a confluent Vandermonde matrix at
the distinct nodes $\alpha_1,\ldots,\alpha_s$, whose determinant is
non-zero.

A Galois automorphism $\sigma \in
\mathrm{Gal}(\mathbb{K}/\mathbb{Q})$ permutes the characteristic
roots, and thus when applied to the above determinant, will
permute its columns. As a result, for any such $\sigma$,
$\sigma(\Delta) = \pm \Delta$, and therefore the quantity $\Delta^2$
is stable under Galois automorphisms. We conclude that $\Delta^2$
must be a rational number, and since it is also by construction an
algebraic integer,\footnote{Note that every entry of the determinant is an
  algebraic integer.} we must have $\Delta^2 \in \mathbb{Z}$.

Solving with Cramer's rule for the coefficients of $Q_i(x)$ gives, via
the Hadamard inequality again,
that the coefficients are bounded by $kA |\Delta|$.  Thus, replacing
${\boldsymbol u}$ by $\Delta^2 {\boldsymbol u}$, we have that
\begin{equation}
\label{eq:bound:coef}
Q_i(x):=\sum_{j=0}^{\nu_i-1} c_{i,j}x^j ,\quad{\text{\rm where }}
|c_{i,j}|\le (k^2A)^{2k^2+1}\quad {\text{\rm and}}\quad c_{i,j}\in
{\mathcal O}_{\mathbb K} \,  .
\end{equation}
From Inequality~\eqref{eq:bound:coef} and Item~1 of
Proposition~\ref{thm:height_bounds} we have the height bound
\begin{equation}
  h(c_{i,j}) \leq (2k^2+1) \log (k^2 A)  \quad\text{for all } i\in
  \{1,\ldots,s\}
  \text{ and }j   \in \{0,\ldots,\nu_i-1\}\,  .
\label{eq:height_coef}
\end{equation}

\textbf{Step 2: Reduction modulo~$P$.}  Fix $n \in \mathcal{S}(X)$
such that $u_n=0$ and consider a representation $n=qP+a$.  Let
$\mathfrak{p}$ be a prime ideal factor of $P$ in
${\mathcal O}_{\mathbb K}$, and recall that there is an automorphism
$\sigma\in {\rm Gal}({\mathbb K}/{\mathbb Q})$ such that
$\sigma(\alpha) \equiv \alpha^P \bmod\mathfrak{p}$ for all
$\alpha \in \mathcal{O}_{\mathbb{K}}$ (see~\cite[Chap.~III.4, Theorem
27]{frohlich1991algebraic}).\footnote{In the case that $\mathbb K=\mathbb Q$, this is just Fermat’s Little Theorem.}
When $\sigma$ is unique with this
property, it is called the \emph{Frobenius automorphism} associated
with $\mathfrak p$.  In case $\sigma$ is not unique (namely, when the
prime $P$ ramifies in $\mathbb K$), we choose $\sigma$ arbitrarily,
and, by a slight abuse of terminology, nevertheless refer to our choice as a
Frobenius automorphism.

From $u_n=0$ and $n=qP+a$ we
have
 \begin{equation}
 \label{eq:1}
 \sum_{i=1}^s Q_i(a) \alpha_i^a \sigma(\alpha_i)^q  \equiv 0 \pmod
 {\mathfrak{p}} \, .
\end{equation}

Recall that our goal is to establish~\eqref{eq:desired}.
If $X \leq  \exp(10k^2\log k)$ then clearly~\eqref{eq:desired} holds.
We proceed now under the assumption
$X> \exp(10k^2\log k)$.  Since $a \in B(X)$ and interval $B(X)$ has
left endpoint $\frac{\log X}{\sqrt{\log_3 X}}$, we have
$a \geq \frac{\log X}{\sqrt{\log_3 X}}>4k^2+3$.  It follows that  $a^k<k^a$.
Noting also that $q \leq a$,
since $q \in A(X)$ and $a \in B(X)$,
the absolute value of
 the left-hand side of~\eqref{eq:1} is at most
\[
 (k^2A)^{2k^2+1}\,k\,a^k\rho^{2a}
 \;<\;(k^2A)^{2k^2+1}\,k^{a+1}(kA)^{2a}
 \;\le\;(kA)^{4k^2+3}(kA)^{3a}
 \;<\;(kA)^{5a},
\]
using $a^k<k^a$ (valid since $a>4k^2+3$), $\rho<kA$, and
$4k^2+3<2a$.

 Write $\beta$ for  the left-hand
 side of~\eqref{eq:1} and suppose that $\beta\neq 0$.  Then $\beta$ is a non-zero algebraic integer of
 degree at most $k!$, all of whose conjugates have absolute value at
 most $(kA)^{5a}$, and which is divisible by $\mathfrak{p}$.
This implies that $P$ divides $N_{\mathbb K/\mathbb Q}(\beta)$, which is an integer of  size at most
 $
 (kA)^{5ak!}.
 $
 Since \[P\ge \frac{X-a}{q}\ge \frac{X-\log X}{{\sqrt{\log X}}}>\sqrt{X}\]
 for $X>X_0:=100$, and $a \leq 2\log X/\sqrt{\log_3 X}$, taking logs
 we have
\begin{gather*}
 5k! \log (kA) \frac{2\log X}{\sqrt{\log_3 X}}>\frac{\log X}{2}.
\end{gather*}
It follows that $ {\sqrt{\log_3 X}}<20k!\log(kA)$, and so
$ X<\exp_3((20k!\log(kA))^2)$.
By Lemma~\ref{lem:absorb}, this yields
$2X\le\mathcal T$, which is Inequality~\eqref{eq:desired}.

\textbf{Step 3: Companion equation.}  In Step 2 we have proved
Inequality~\eqref{eq:desired} under the assumption that the left-hand
side of~\eqref{eq:1} is non-zero for some representation of $n$.  Now
suppose, on the contrary, that the left-hand side of~\eqref{eq:1} is
zero for each of the $r(n)> \log_4 X$ representations of $n$.  Of
these representations, at least $(\log_4 X)/k!$ have the same
Frobenius automorphism $\sigma$.  For this choice of $\sigma$ we have
that the \emph{companion equation} (the equational analog of the
congruence~\eqref{eq:1})
\begin{gather}
  \sum_{i=1}^s Q_i(a) \alpha_i^a\sigma(\alpha_i)^q=0
\label{eq:comp}
\end{gather}has at least $(\log_4 X)/k!$ solutions in integer variables
$q,a$.  The remainder of the proof is dedicated to deriving an upper
bound on the number of solutions of~\eqref{eq:comp} that arise from
representations of $n$. From this we obtain~\eqref{eq:desired}, as
desired.

Every solution of~\eqref{eq:comp} has a non-degenerate vanishing
sub-sum and we focus on bounding the number of such sub-sums.  The
following claim, proven in Section~\ref{sec:vanishing}, considers
sub-sums that involve only terms from a single summand
$Q_i(a)\alpha_i^a\sigma(\alpha_i)^q$ of~\eqref{eq:comp}.
\begin{restatable}{claim}{VANISHING}
  Suppose $R_i(a)\alpha_i^a\sigma(\alpha_i)^q=0$, where $R_i$ is a
  sub-polynomial of $Q_i$ (i.e., every monomial in $R_i$ appears in
  $Q_i$) for some $i\in \{1,\ldots,s\}$.  Then
  $2X\le\mathcal T$.
  \label{claim:vanishing}
\end{restatable}

Since the conclusion of Claim~\ref{claim:vanishing} is
precisely Inequality~\eqref{eq:desired}, it remains to bound the total
number of non-degenerate solutions of each of the at most $2^k$
sub-equations of the form
 \begin{equation}
 \label{eq:4}
 \sum_{i\in I} R_i(a)\sigma(\alpha_i)^q\alpha_i^a=0 \, , 
\end{equation}
of~\eqref{eq:comp}, where $I\subseteq \{1,\ldots,s\}$ contains
at least two elements, and where $R_i(x)$ is a sub-polynomial of
$Q_i(x)$ for all $i\in I$.  For this task, a key structure is the group
$\mathcal P$ of ${\boldsymbol z}=(z_1,z_2)\in {\mathbb Z}^2$ such that
 \begin{gather*}
\sigma(\alpha_i)^{z_1}\alpha_i^{z_2}=\sigma(\alpha_j)^{z_1}\alpha_j^{z_2}\quad{\text{\rm
    for~all}}\quad i,j\in I \, . 
\end{gather*}
For ${\boldsymbol z} =(z_1,z_2) \in \mathcal{P}$ we have
$\sigma(\alpha_i/\alpha_j)^{z_1}=(\alpha_j/\alpha_i)^{z_2}$ for all
$i,j\in I$.  As shown in Section~\ref{sec:background}, since
$\alpha_i/\alpha_j$ is not a root of unity, this entails that $z_1= z_2$ or
$z_1=-z_2$.  There are thus three possibilities for $\mathcal{P}$: either
$\mathcal{P}=\{\boldsymbol{0}\}$, $\mathcal{P}$ is parallel to $(1,1)$,
or $\mathcal{P}$ is parallel to $(1,-1)$.

\textbf{Step 4: The easy case.}
This is the case ${\mathcal P}=\{{\bf 0}\}$.
Recall that in Equation~\eqref{eq:4} polynomial $R_i$ has degree at
most $\nu_i-1$.   Now
 Theorem~\ref{thm:SchSch} shows that if we put
 \begin{gather*}
 A_0:=\sum_{i\in I} \binom{2+\nu_i-1}{2} \quad\text{and}\quad
 B_0=\max\{2,A_0\} \,
 \end{gather*}
then the number of solutions $(a,q)$ of~\eqref{eq:4} is at most $
2^{35B_0^3} (k!)^{6B_0^2}$.  But
 \begin{gather*}
 A_0\le \sum_{i\in I} \binom{\nu_i+1}{2}=\sum_{i\in I}\frac{\nu_i(\nu_i+1)}{2}\le \sum_{i\in I} \nu_i^2\le k^2,
 \end{gather*}
 so the number of solutions of~\eqref{eq:4} is at most
 $
 2^{35k^6} (k!)^{6k^4}.
 $
After multiplying by $2^k$ to account for all non-degenerate
sub-sums, the total number of companion-equation solutions arising
from easy sub-equations is at most
\[
 E_k:=2^k2^{35k^6}(k!)^{6k^4}.
\]
If $(\log_4X)/k!\leq2E_k$, then the same calculation as above gives
\begin{equation}
 \label{eq:max}
 X<\max\{\exp_5(10^{10}),\exp_5(25k^6)\},
\end{equation}
which yields Inequality~\eqref{eq:desired}.  We may therefore suppose
in Step~5 that $(\log_4X)/k!>2E_k$.  In that case the hard
sub-equations account for more than $(\log_4X)/(2k!)$ of the
companion-equation solutions.

\textbf{Step 5: The hard case.}  We are left with the case
${\mathcal P}\neq\{{\bf 0}\}$, where we cannot apply
Theorem~\ref{thm:SchSch}.  In this case $\mathcal P$ is either a
subgroup of $\{(z,z):z\in\mathbb Z\}$ or a subgroup of
$\{(z,-z):z\in\mathbb Z\}$.

Fix one of the hard sub-equations~\eqref{eq:4}.  If
$(m,m)\in\mathcal P$ is non-zero, then, after fixing $i_0\in I$,
\[
 \left(\frac{\sigma(\alpha_i)\alpha_i}
 {\sigma(\alpha_{i_0})\alpha_{i_0}}\right)^m=1
 \qquad(i\in I).
\]
If instead $(m,-m)\in\mathcal P$ is non-zero, the same assertion holds
with $\sigma(\alpha_i)/\alpha_i$ in place of
$\sigma(\alpha_i)\alpha_i$.  Consequently, in either case there are a
non-zero $\lambda\in\mathbb K$, roots of unity
$\zeta_i\in\mathbb K$, and a sign $\eta\in\{\pm1\}$ such that
\[
 \sigma(\alpha_i)=\lambda\zeta_i\alpha_i^\eta
 \qquad(i\in I).
\]
Here $\eta=-1$ in the diagonal case and $\eta=1$ in the
anti-diagonal case.  Let $M:=|\mu(\mathbb K)|$ be the number of roots
of unity in $\mathbb K$.  Since $\mu(\mathbb K)$ is cyclic,
$\varphi(M)\leq[\mathbb K:\mathbb Q]\leq k!$; the elementary bound
$M\leq2\varphi(M)^2$ therefore gives
\begin{gather}
 M\leq2(k!)^2.
 \label{eq:roots-unity-bound}
\end{gather}
After cancelling $\lambda^q$ in~\eqref{eq:4} and restricting $q$ to a
fixed residue class $r\pmod M$, we obtain
\begin{equation}
\label{eq:5}
 \sum_{i\in I}\widetilde R_{i,r}(a)\alpha_i^{a+\eta q}=0,
 \qquad \widetilde R_{i,r}:=\zeta_i^rR_i.
\end{equation}
Multiplication by a root of unity preserves the support, degree,
algebraic integrality, and all height bounds for the coefficients of
$R_i$.  We may therefore apply the argument below separately in each
of the at most $M$ residue classes and, within a fixed class, suppress
the tildes and the subscript $r$.

As in Step 4, we will establish the
bound~\eqref{eq:desired} by
giving an upper bound on the number of solutions of~\eqref{eq:5} and
hence on the number of representations of $n$.  In lieu of
Theorem~\ref{thm:SchSch}, we use a bespoke argument that uses ideas of
\cite{AV09,ESS,SchlickeweiS00}, but which is greatly
simplified by exploiting the assumption that no two representations of
$n$ are correlated.
For either fixed sign $\eta$, the map $(q,a)\mapsto a+\eta q$ is
injective on the representations under consideration, and any two
distinct values in its image differ by at least $\sqrt{\log X}$;
otherwise the corresponding representations would be correlated.

The argument is by induction on the number of summands $|I| \le k$.
To get started, we set $\ell:=|I|$, relabel the roots so that
$I=\{1,\ldots,\ell\}$, and restate Equation~\eqref{eq:5} as follows:
\begin{gather*}
\sum_{i=1}^{\ell} R_i(a)\alpha_i^{a+\eta q}=0 \, .
\end{gather*}
We dehomogenise the above equation by dividing by the first
summand, yielding 
\begin{equation}
\label{eq:6}
1=\sum_{i=2}^{\ell} (-R_i(a)/R_1(a))(\alpha_i/\alpha_1)^{a+\eta q}.
\end{equation}

Our goal is to apply Theorem~\ref{thm:AV09} in order to find a
homogeneous linear relation among the summands on the right-hand side
of~\eqref{eq:6}.  This will yield an equation similar to~\eqref{eq:5}
but with strictly fewer summands.  To this end, let $\Gamma$ be the
rank-one multiplicative subgroup of
$\left(\mathbb{K}^\ast\right)^{\ell-1}$ generated by
$\boldsymbol{\gamma}:= (\alpha_2/\alpha_1,\ldots,\alpha_\ell/\alpha_1)$. Then
Equation~\eqref{eq:6} can be written in the form
$ {\boldsymbol x}^\top {\boldsymbol y}=1$, where
\begin{equation}
\begin{array}{rcl}
{\boldsymbol x}&\, :=\, & \boldsymbol{\gamma}^{a+\eta q} \\
{\boldsymbol y}&\, :=\, &-(R_2(a)/R_1(a),\ldots,R_\ell(a)/R_1(a)) \, .
\end{array}
  \label{eq:eps}
\end{equation}
  Recalling that $h$ denotes absolute
logarithmic Weil height, we have the following claim, which is
proven in Section~\ref{sec:technical}.
\begin{restatable}{claim}{TECHNICAL}
  For $\boldsymbol x,\boldsymbol y$ in~\eqref{eq:eps} and 
  $\varepsilon:=(8k)^{-6k^3}$, if
  $2X > \mathcal T$ then
  $h({\boldsymbol y})\le (1+h({\boldsymbol x}))\varepsilon$.
\label{claim:heightbound}
\end{restatable}

Since the inequality $2X\leq \mathcal T$
implies the desired bound~\eqref{eq:desired},
by Claim~\ref{claim:heightbound} we may assume without loss of generality that
$h({\boldsymbol y})\le (1+h({\boldsymbol x}))\varepsilon$.  This height inequality
allows us to apply Theorem~\ref{thm:AV09} to conclude that there is a
collection of at most $(8k)^{6k^3(k+1)}$ vectors
${\bf A}=(A_2,\ldots,A_\ell)$ in ${\overline{{\mathbb Q}}}^{\ell-1}$
such that each solution of \eqref{eq:6} satisfies
\begin{gather}
\sum_{i=2}^{\ell} A_i R_i(a) \alpha_i^{a+\eta q}=0
\label{eq:REP}
\end{gather}
for one such $\bf A$.  We will use such linear relations to
proceed by induction.  

Fix a vector ${\bf A}=(A_2,\ldots,A_{\ell})$ among the
$(8k)^{6k^3(k+1)}$ possibilities and consider the corresponding version
of Equation~\eqref{eq:REP}.
Assume that at least three of the $A_i$'s are nonzero and re-index so
that the non-zero $A_i$'s have indices
$i=2,3,\ldots,\ell'$, where $\ell' \le \ell$.  We
dehomogenise the above equation to get
\begin{gather*}
1=\sum_{i=3}^{\ell'} (-A_i/A_2) (R_i(a)/R_2(a)) (\alpha_i/\alpha_2)^{a+\eta q}.
\end{gather*}
We now take
$\Gamma\subseteq \left(\mathbb{K}^\ast\right)^{\ell'-2}$ to
be the multiplicative subgroup (of rank at most two) generated by the vectors
$(\alpha_3/\alpha_2,\ldots,\alpha_{\ell'}/\alpha_2)$ and
$((-A_3/A_2),\ldots,(-A_{\ell'}/A_2))$. The above equation is again of
the form ${\boldsymbol x}^\top{\boldsymbol y}=1$, where now
\begin{equation}
  \begin{array}{rcl}
{\boldsymbol x}&\,:=\, &((-A_3/A_2) (\alpha_3/\alpha_2)^{a+\eta
                        q},\ldots,(-A_{\ell'}/A_2)
                        (\alpha_{\ell'}/\alpha_2)^{a+\eta q})) \\
    {\boldsymbol y} &:=& (R_3(a)/R_2(a),\ldots, 
                         R_{\ell'} (a)/R_2(a)) \, .
                         \end{array}
\label{eq:induct}
\end{equation}
To continue the induction we wish to establish
the height bound 
\begin{equation}
\label{eq:11}
h({\boldsymbol y})\le (1+h({\boldsymbol x}))\varepsilon \, ,
\end{equation}
with $\boldsymbol x,\boldsymbol y$ as in~\eqref{eq:induct} (
still with $\varepsilon=(8k)^{-6k^3}$).   The
challenge is that the components of ${\bf A}$ (arising from the application
of Theorem~\ref{thm:AV09}) are not known.  But in the case at hand
this is not a problem thanks to the following lemma, which is proved in
Section~\ref{sec:easy}:
\begin{restatable}{claim}{EASY}
  For $\boldsymbol x,\boldsymbol y$ in~\eqref{eq:induct} and
  $\varepsilon=(8k)^{-6k^3}$, if
  $2X>\mathcal T$,
  then there is at most one value of $a+\eta q$
  in~\eqref{eq:induct} for which
  $h({\boldsymbol y})\leq(1+h({\boldsymbol x}))\varepsilon$ fails.
\label{claim:easy}
\end{restatable}
If $2X\leq\mathcal T$, then
Inequality~\eqref{eq:desired} already holds.  We may therefore assume
the opposite inequality and apply Claim~\ref{claim:easy}.  It follows
that at most one solution of~\eqref{eq:REP} has a corresponding vector
$\boldsymbol x$ that fails to satisfy~\eqref{eq:11}; all remaining
solutions may be passed to the next induction step.

Within a fixed residue class modulo $M$, assign each hard solution
to one of the at most $2^k$ hard sub-equations of~\eqref{eq:comp}
that it satisfies, and organise the induction as a forest with these
sub-equations as roots.  At the $i$-th application of
Theorem~\ref{thm:AV09}, the multiplicative group has rank at most
$i$.  Assign every non-exceptional solution to one of the child
equations that it satisfies.  The number of nodes after $j$
applications is then at most
\[
 2^k\prod_{i=1}^j(8k)^{6k^3(k+i)}.
\]
At every non-initial node, Claim~\ref{claim:easy} allows at most one
exponent value to be stranded rather than passed to a child.  If no
two-term leaf contains two distinct exponent values, then each leaf
also contains at most one exponent value.  Counting one possible
stranded value at every node, including the roots, therefore gives the
upper bound
\begin{gather*}
 N_k
 \, :=\, 2^k\sum_{j=0}^{k-2}
       \prod_{i=1}^j(8k)^{6k^3(k+i)} \, < \, (8k)^{10k^5}
\end{gather*}
for the total number of solutions in the forest; here the product for
$j=0$ is empty.  Indeed,
\[
 \prod_{i=1}^j(8k)^{6k^3(k+i)}
 =(8k)^{6k^3\left(jk+j(j+1)/2\right)},
\]
and, for $0\leq j\leq k-2$, the exponent is at most
\[
 3k^3(k-2)(3k-1)\leq9k^5.
\]
The remaining factor $2^k(k-1)$ is at most $(8k)^{k^2}$ for $k\geq2$,
and $9k^5+k^2\leq10k^5$.

Consequently, more than $N_k$ solutions in a fixed residue class force
a two-term leaf containing two distinct exponent values.  Thus, for
some $i\ne j$, non-zero $A_i,A_j$, and $\eta\in\{\pm1\}$, one has
\[
 A_iR_i(a)\alpha_i^{a+\eta q}
 +A_jR_j(a)\alpha_j^{a+\eta q}=0,
\]
and the same equation holds with $(q,a)$ replaced by $(q',a')$, where
$a+\eta q\ne a'+\eta q'$.
Dividing
these two equations yields 
\begin{gather}
\frac{R_i(a)}{R_i(a')} \, \frac{R_j(a')}{R_j(a)} =\left(\frac{\alpha_j}{\alpha_i}\right)^{(a+\eta q)-(a'+\eta q')}.
\label{eq:take-heights}
\end{gather}

We now prepare to take heights in Equation~\eqref{eq:take-heights}.
Under our standing assumption that~\eqref{eq:desired} fails, the first
case in the proof of Claim~\ref{claim:heightbound} shows that both
$h(a)$ and $h(a')$ are at least $3k^2\log(k^2A)$.  Hence
Equation~\eqref{eq:XX} applies to both $a$ and $a'$.  Since
$a,a'\in B(X)$, for sufficiently large $X$ we have
$h(a),h(a')\leq\log_2X$, and therefore
\[
 h(R_i(a)/R_i(a'))
 \leq h(R_i(a))+h(R_i(a'))
 \leq2k\bigl(h(a)+h(a')\bigr)
 \leq4k\log_2X
\]
for $i=1,\ldots,\ell$.  Meanwhile, by~\eqref{eq:dob},
\[
 h(\alpha_i/\alpha_j)\geq
 \frac{2}{k^2(\log(3k^2))^3}.
\]
Since $|(a+\eta q)-(a'+\eta q')|\geq\sqrt{\log X}$, we have
\begin{gather*}
  8k \log_2 X>{\sqrt{\log X}} \,  h(\alpha_i/\alpha_j)>
\frac{2\sqrt{\log X}}{k^2(\log(3k^2))^3} \, .
\end{gather*}
The above inequality can be rewritten $\sqrt{\log X}<c \log_2 X$ where
$c:=4(k\log(3k^2))^3$.  Since $\log c\le20k^4$,
Corollary~\ref{cor:elementary} yields
$2X<\exp_5(10^{10}k^6)\le\mathcal T$, which implies
Inequality~\eqref{eq:desired}, our ultimate goal.

It remains to consider the case in which every residue class modulo
$M$ contributes at most $N_k$ solutions to the hard sub-equations.
Using~\eqref{eq:roots-unity-bound} and the estimate for $N_k$, we then
have
\[
 \frac{\log_4X}{2k!}
 <MN_k
 \leq2(k!)^2(8k)^{10k^5}.
\]
Thus
\[
 \log_4X<4(k!)^3(8k)^{10k^5}.
\]
For $k\geq2$,
\[
 \log4+3\log(k!)+10k^5\log(8k)
 \leq55k^5\log k,
\]
and consequently
\[
 X<\exp_5(55k^5\log k).
\]
Since $55k^5\log k<10^{10}k^6$ for $k\geq2$, this again implies
Inequality~\eqref{eq:desired}.
This concludes the proof of Theorem~\ref{thm:1}.

\section{The Density of \texorpdfstring{${\mathcal S}$}{S}}
\label{sec:density}
This section is devoted to a proof of the following result.
\begin{theorem}
  The set $\mathcal S$ has lower density at least $1/8$
  unconditionally and has density $1$ subject to
  Martin's uniform formulation of the Bateman--Horn conjecture
  (Hypothesis UH).
  \label{thm:density}
\end{theorem}

Recall from Section~\ref{sec:2} that we exclude from $\mathcal S$ all
$n\in\mathbb{N}$ that have two correlated representations.  In
Section~\ref{sec:count0} we show that the set of numbers thus excluded
has density zero.  In Section~\ref{sec:count} we show,
assuming Hypothesis UH, that
\begin{gather}
  \# \left\{ n \in [X,2X]: r(n)>\log_4  X\right\} = (1+o(1)) X \, . 
\label{eq:ultimate}
\end{gather} 
We also show unconditionally that
\begin{gather}
  \# \left\{ n \in [X,2X]: r(n)>\log_4 X\right\}
  \geq ((1/\kappa)+o(1))X,
\label{eq:ultimate2}
\end{gather}
where $\kappa=8$ is the absolute constant in
Inequality~\eqref{eq:BRUN}.  Together with
Lemma~\ref{lem:far}, these estimates give the corresponding local
bounds for $|\mathcal S(X)|$.  At the end of the section we pass from
these local bounds to the density of $\mathcal S$.

In this section the indices $p,q,P,P'$ in summations and
products  run over positive primes.

The following elementary parametrisation will be used twice.

\begin{lemma}
\label{lem:param}
Let $q\ne q'$ be primes in $A(X)$, let $a,a'\in B(X)$ with
$d:=a-a'\ne0$, and let $P_0$ be the unique integer with
$0\leq P_0<q'$ and $qP_0\equiv-d\pmod{q'}$.  Put
$P'_0:=(qP_0+d)/q'\in\mathbb Z$, so that $qP_0-q'P'_0=-d$ and
\[
 |P_0|<q',\qquad |P'_0|\leq q+\frac{|d|}{q'}\ll\log X.
\]
Then the integer solutions $(P,P')$ of
\[
 qP+a=q'P'+a'\in[X,2X]
\]
are exactly the pairs $\bigl(f_1(t),f_2(t)\bigr)$ with
\[
 f_1(t):=P_0+q't,\qquad f_2(t):=P'_0+qt,
\]
where $t$ ranges over the integers of the interval $[T_1,T_2]$,
$T_j:=(jX-a-qP_0)/(qq')$.  In particular $T_2-T_1=X/(qq')$ and
$T_1\asymp T_2\asymp X/(qq')\gg X/\log X$; the forms $f_1,f_2$ have
coefficients $O(\log X)$ and coprime leading coefficients; and, in the
notation of~\eqref{eq:sieve}, $E=qq'|d|\ne0$.
\end{lemma}
\begin{proof}
The congruence $qP\equiv a'-a\pmod{q'}$ determines $P$ modulo $q'$,
and $qP_0\equiv-d\pmod{q'}$ by construction, so $P=P_0+q't$ for some
$t\in\mathbb Z$; then $P'=(qP+d)/q'=P'_0+qt$.  The membership
$qP+a\in[X,2X]$ is equivalent to $t\in[T_1,T_2]$.  The remaining
assertions are immediate, using $q,q'\le\sqrt{\log X}$, $|d|\le2\log
X/\sqrt{\log_3X}$ and $E=|q'q(q'P'_0-qP_0)|=qq'|d|$.
\end{proof}

\subsection{Counting correlated representations}
\label{sec:count0}
We will need the following simple fact:
\begin{proposition}
$\sum_{q \in A(X)} \frac{1}{q} \sim \log_3 X$.
\label{prop:inv-bound}
\end{proposition}
\begin{proof}
  
  By Mertens's theorem (see, e.g.,
  \cite[Theorem~427]{hardywright54}) we have $\sum_{q \leq n}
\frac{1}{q} = \log \log  n + O(1)$.
Since $A(X) = [\log_2 X,\sqrt{\log X}]$, it follows that
$\sum_{q \in A(X)} \frac{1}{q}=\log\log (\sqrt{\log X}) - \log_4 X+O(1)$,
which directly gives the result.
\end{proof}

We now have:
\begin{lemma}
\label{lem:far}
The set of $n\in[X,2X]$ having two distinct representations
$n=Pq+a=P'q'+a'$ for which
\[
 |(a+\eta q)-(a'+\eta q')|<\sqrt{\log X}
\]
for some $\eta\in\{\pm1\}$ has cardinality
$O(X/(\log X)^{1/3})$.
\end{lemma}
\begin{proof}
We split the count according to whether $q=q'$ and whether $a=a'$.
First suppose that $q\ne q'$ and $a\ne a'$.  For fixed such parameters,
we count pairs of primes $P,P'$ satisfying
\begin{gather}
 qP+a=q'P'+a'\in[X,2X].
 \label{eq:corr}
\end{gather}
By Lemma~\ref{lem:param}, these pairs are
$\bigl(f_1(t),f_2(t)\bigr)$ for integers $t$ in an interval of length
$X/(qq')$ contained in $[0,T]$ with $T\asymp X/(qq')$, and
$E=qq'|d|\ne0$ with $\gcd(q',q)=1$.  Hence the sieve
bound~\eqref{eq:sieve} gives
\[
 \#\{t:f_1(t),\ f_2(t)\text{ are prime}\}
 \ll\frac{E}{\varphi(E)}\cdot\frac{X}{qq'(\log X)^2}
 \ll\frac{X\log_3X}{qq'(\log X)^2},
\]
using $m/\varphi(m)\ll\log\log m$
(see~\cite[Theorem~328]{hardywright54}).  The correlation inequality
and $q,q'\leq\sqrt{\log X}$ imply
$|a-a'|<2\sqrt{\log X}$.  Summing over the choices of $a,a',q,q'$ and
using Proposition~\ref{prop:inv-bound}, we obtain
\[
 \ll\frac{X\log_3X}{(\log X)^2}
 \left(\sum_{q\in A(X)}\frac1q\right)^2
 \frac{\log X\sqrt{\log X}}{\sqrt{\log_3X}}
 \ll\frac{X(\log_3X)^{5/2}}{\sqrt{\log X}}
 =O\!\left(\frac{X}{(\log X)^{1/3}}\right).
\]

Next suppose that $q=q'$.  Since the representations are distinct,
$a\ne a'$, and the equality of the representations shows that
$q\mid a'-a$.  Put $h:=(a'-a)/q\ne0$.  Then
$P'=P-h$, while correlation is equivalent to
$|a-a'|<\sqrt{\log X}$, so $|h|<\sqrt{\log X}/q$.  For fixed $q,a,h$,
the sieve bound~\eqref{eq:sieve}, applied to the forms $t$ and $t-h$
on an interval of length $\ll X/q$, gives
\[
 \#\{P:P,\ P-h\text{ are prime}\}
 \ll\frac{X}{q(\log X)^2}\frac{|h|}{\varphi(|h|)}
 \ll\frac{X\log_3X}{q(\log X)^2}.
\]
There are $O(\log X/\sqrt{\log_3X})$ choices for $a$ and
$O(\sqrt{\log X}/q)$ choices for $h$.  Hence the total contribution in
this case is
\[
 \ll\frac{X\sqrt{\log_3X}}{\sqrt{\log X}}
 \sum_{q\in A(X)}\frac1{q^2}
 \ll\frac{X\sqrt{\log_3X}}
 {\sqrt{\log X}\,\log_2X}
 =O\!\left(\frac{X}{(\log X)^{1/3}}\right).
\]

Finally, suppose that $a=a'$ and $q\ne q'$.  Then $qP=q'P'$, and
unique factorisation of a product of two primes gives $P=q'$ and
$P'=q$.  Thus $n=qq'+a\ll\log X$, which is incompatible with
$n\in[X,2X]$ for all sufficiently large $X$.  The case $a=a'$ and
$q=q'$ would give the same representation and therefore does not
occur.  Combining the three cases proves the lemma.
\end{proof}

\subsection{Counting all representations}
\label{sec:count}
In view of Lemma~\ref{lem:far}, to show that $\mathcal S$ has density
$1$ it
suffices to establish Equation~\eqref{eq:ultimate}.
We will use the Cauchy-Schwarz inequality.  To set this up, for $i\in \{0,1,2\}$ write
\[ M_i(X) = \sum_{\substack{n\in[X,2X]\\r(n)>\log_4 X}} r(n)^i \, .\]

We estimate the first moment $M_1(X)$ as follows.  We have 
\begin{align}
\sum_{n\in [X,2X]}r(n)&=\sum_{\substack{q\in A(X) \\ a\in
  B(X)}}\,\,\sum_{\frac{X-a}{q}\le P\le \frac{2X-a}{q}}1 \notag \\
&=(1+o(1))\sum_{\substack{q\in A(X) \\ a\in B(X)}}\frac{X}{q\log{X}}
  \quad\text{(by the Prime Number Theorem)} \notag \\
&=(1+o(1))X{\sqrt{\log_3 X}} \quad\text{(by
                                                        Proposition~\ref{prop:inv-bound}).}
                                                        \label{eq:XXX}
\end{align}
It follows immediately that $M_1(X)=(1+o(1))X\sqrt{\log_3 X}$.

Turning to the second moment $M_2(X)$,
if we were able to show that 
\begin{equation}
\sum_{n\in[X,2X]}r(n)^2=(1+o(1))X \log_3 X 
\label{eq:SecondMoment}
\end{equation}
then it would follow that $M_2(X) = (1+o(1)) X \log_3 X$ and hence,
by the Cauchy-Schwarz inequality $M_0(X)M_2(X) \geq M_1(X)^2$, that
$M_0(X) = (1+o(1))X$, which would establish~\eqref{eq:ultimate}.

It thus suffices to establish Equation~\eqref{eq:SecondMoment}.  This
is out of the reach of unconditional techniques, but
it follows from Hypothesis UH, stated in Conjecture~\ref{conj:bateman-horn}.

Fix $a\ne a'\in B(X)$ and $q\ne q'\in A(X)$, put $d:=a-a'$, and suppose
that $\gcd(d,qq')=1$ and $2\mid d$.  We consider pairs of primes $P,P'$
such that
\begin{gather}
  qP+a=q'P'+a' \in [X,2X].
\label{eq:CC}
\end{gather}
Let $P_0,P'_0,f_1,f_2,T_1,T_2$ be as in
Lemma~\ref{lem:param}.  By that lemma, the number of solutions of
Equation~\eqref{eq:CC} in primes $P$ and $P'$ is, up to an $O(1)$
endpoint error, the number of integers $t\in[T_1,T_2]$ for which
$f_1(t)$ and $f_2(t)$ are both prime; moreover $T_2-T_1=X/(qq')$,
$T_1\gg X/\log X$, the coefficients of $f_1,f_2$ are $O(\log X)$, and
the coefficients of $f=f_1f_2$ are $(\log X)^{O(1)}$.

Set
\[
C:=2\prod_{p>2}\frac{p(p-2)}{(p-1)^2}\approx1.32,
\qquad
g(m):=\prod_{\substack{p\mid m\\p>2}}\frac{p-1}{p-2}.
\]
For sufficiently large $X$, the primes $q,q'$ are odd.  We record the
local-factor calculation explicitly.

\begin{claim}
\label{claim:singular-series}
Under the standing assumptions $q\ne q'$, $\gcd(d,qq')=1$, and
$2\mid d$, the pair $f_1,f_2$ is admissible and, for $f:=f_1f_2$,
its Bateman--Horn constant is
\[
C_f=Cg(q)g(q')g(|d|)=Cg(qq'|d|).
\]
\end{claim}
\begin{proof}
Let $\omega_f(p)$ denote the number of roots of $f$ modulo $p$.  If
$p>2$ and $p\nmid qq'$, then each of $f_1$ and $f_2$ has one root modulo $p$.
These roots coincide if and only if
$qP_0\equiv q'P'_0\pmod p$, which, since
$qP_0-q'P'_0=-d$, is equivalent to $p\mid d$.  Hence
$\omega_f(p)=1$ when $p\mid d$ and $\omega_f(p)=2$ otherwise.

At $p=q$, the form $f_1$ has one root, whereas $f_2$ is a non-zero
constant: indeed $q'P'_0\equiv d\not\equiv0\pmod q$.  Thus
$\omega_f(q)=1$.  Similarly, $\omega_f(q')=1$.  Finally, since $q,q'$
are odd and $d$ is even, the relation $qP_0-q'P'_0=-d$ shows that
$P_0\equiv P'_0\pmod2$; the two roots modulo $2$ therefore coincide,
so $\omega_f(2)=1$.  In particular, $\omega_f(p)<p$ for every prime
$p$, and the pair is admissible.

Relative to the generic odd-prime local factor
$p(p-2)/(p-1)^2$, each odd prime for which $\omega_f(p)=1$ contributes
the correction factor
\[
\frac{p/(p-1)}{p(p-2)/(p-1)^2}=\frac{p-1}{p-2}.
\]
The sets of odd primes dividing $d$, equal to $q$, and equal to $q'$
are disjoint by the hypotheses.  Including the local factor $2$ at
$p=2$, we obtain
\[
C_f=2\prod_{p>2}\frac{p(p-2)}{(p-1)^2}
\,g(q)g(q')g(|d|)=Cg(q)g(q')g(|d|),
\]
as claimed.
\end{proof}

We apply Conjecture~\ref{conj:bateman-horn} at the two endpoints
$T_1,T_2$.  Since the coefficients of $f$ are $(\log X)^{O(1)}$ and
$T_1,T_2\gg X/\log X$, Martin's coefficient condition is satisfied
uniformly for all the quadruples under consideration.  Moreover, for
$t\in[T_1,T_2]$,
\[
\log f_1(t)=\log X+O(\log_2 X),\qquad
\log f_2(t)=\log X+O(\log_2 X),
\]
uniformly in $a,a',q,q'$.  Hence
\[
\operatorname{li}(f;T_2)-\operatorname{li}(f;T_1)
=\frac{X}{qq'(\log X)^2}
+O\!\left(\frac{X\log_2 X}{qq'(\log X)^3}\right).
\]
It follows that the number of prime pairs satisfying~\eqref{eq:CC} is
\begin{gather}
\frac{C_fX}{qq'(\log X)^2}
+O\!\left(\frac{C_fX\log_2 X}{qq'(\log X)^3}+1\right),
\label{eq:est}
\end{gather}
uniformly in $a,a',q,q'$.  By Claim~\ref{claim:singular-series}, the
main term in~\eqref{eq:est} is
\[
C\frac{X}{qq'(\log X)^2}g(q)g(q')g(|a-a'|).
\]

For completeness, if $2\nmid d$, then any prime solution has
$P=2$ or $P'=2$.  If $\gcd(d,qq')>1$, then $q\mid d$ forces $P'=q$,
and $q'\mid d$ forces $P=q'$.  Thus, for every exceptional quadruple
$(a,a',q,q')$, only $O(1)$ prime pairs can occur.  Summed over all
parameters, these exceptional solutions contribute
$O(|A(X)|^2|B(X)|^2)=O((\log X)^3)$, which is absorbed by the error term
used below.

We also record the contributions not covered by the assumptions
$q\ne q'$ and $a\ne a'$.  Identical representations contribute
\[
 \sum_{n\in[X,2X]}r(n)=O\!\left(X\sqrt{\log_3X}\right)
 =o(X\log_3X).
\]
If $q=q'$ but $a\ne a'$, put $h:=(a'-a)/q\ne0$.  Then
$P'=P-h$ and $|h|\ll |B(X)|/q$.  For fixed $q,a,h$, the sieve
bound~\eqref{eq:sieve}, applied to the forms $t$ and $t-h$ on the
relevant interval of length $\asymp X/q$, gives
\[
 \ll\frac{X\log_3X}{q(\log X)^2}
\]
prime pairs.  Summing over $a$, $h$, and $q$ yields the following
bound; the final $O((\log X)^2)$ term absorbs the $O(1)$
contributions from endpoint and locally obstructed cases:
\begin{align*}
 &\ll\frac{X\log_3X}{(\log X)^2}
 |B(X)|^2\sum_{q\in A(X)}\frac1{q^2}+O((\log X)^2)\\
 &\ll X\sum_{q\in A(X)}\frac1{q^2}+O((\log X)^2)
 \ll\frac{X}{\log_2X}+O((\log X)^2)=o(X).
\end{align*}
Finally, if $a=a'$ and $q\ne q'$, then $qP=q'P'$, so unique
factorisation gives $P=q'$ and $P'=q$.  Hence
$n=qq'+a\ll\log X$, which is impossible for $n\in[X,2X]$ once $X$
is sufficiently large.  Thus all of these omitted configurations
contribute $o(X\log_3X)$ in total.

We first check that the uniform error in~\eqref{eq:est} is summable.
For $m\geq1$,
\[
g(m)\ll\frac{m}{\varphi(m)}\ll\log_2(3m)
\quad\text{(see~\cite[Theorem~328]{hardywright54})},
\]
and therefore
\[
\sum_{q\in A(X)}\frac{g(q)}q\ll\log_3 X,
\qquad
\sum_{\substack{a,a'\in B(X)\\a\ne a'}}g(|a-a'|)
\ll\frac{(\log X)^2\log_2 X}{\log_3 X}.
\]
Consequently,
\[
\sum_{\substack{a\ne a'\in B(X)\\q\ne q'\in A(X)\\
2\mid(a-a')\\\gcd(a-a',qq')=1}}
\frac{C_f}{qq'}
\ll(\log X)^2\log_2 X\log_3 X.
\]
The sum of the error terms in~\eqref{eq:est} is thus
\[
O\!\left(
\frac{X(\log_2 X)^2\log_3 X}{\log X}+(\log X)^3
\right)=o(X\log_3 X).
\]
Taking stock, and also including the diagonal and exceptional
contributions, we obtain
\begin{align*}
  \sum_{n\in[X,2X]}r(n)^2
  &=\sum_{\substack{a,a'\in B(X)\\ q,q'\in A(X)}}
  \; \sum_{\substack{P,P'\\ qP+a=q'P'+a'\in[X,2X]}}1  \\
  &=\frac{CX}{(\log X)^2}
  \sum_{\substack{a\ne a'\in B(X)\\ q\ne q'\in A(X)\\ 2\mid(a-a')\\
  \gcd(a-a',qq')=1}}
  \frac{g(q)g(q')g(|a-a'|)}{qq'}
  +o(X\log_3 X).
\end{align*}
We now check that the two new factors $g(q)g(q')$ do not change the
main term.  For a fixed non-zero $d=a-a'$ put
\[
S_d(X):=\sum_{\substack{q\in A(X)\\q\nmid d}}\frac{g(q)}q.
\]
Uniformly for $q\in A(X)$,
\[
g(q)=\frac{q-1}{q-2}=1+O\left(\frac1{\log_2 X}\right),
\qquad
\sum_{q\in A(X)}\frac{g(q)-1}{q}
 \ll\frac1{\log_2 X}.
\]
Moreover, $|d|\leq\log X/\sqrt{\log_3 X}$, and every prime divisor of
$d$ lying in $A(X)$ is at least $\log_2 X$.  It follows that
\[
\#\{q\in A(X):q\mid d\}\leq
\frac{\log|d|}{\log(\log_2 X)}
\leq\frac{\log_2 X}{\log_3 X},
\]
and hence
\[
\sum_{\substack{q\in A(X)\\q\mid d}}\frac{g(q)}q
\ll\frac1{\log_3 X}.
\]
By Proposition~\ref{prop:inv-bound}, therefore,
\[
S_d(X)=(1+o(1))\log_3 X
\]
uniformly for $d=a-a'$ under consideration.  Finally,
\[
\sum_{q\in A(X)}\frac{g(q)^2}{q^2}\ll\frac1{\log_2 X},
\]
so deleting the diagonal $q=q'$ is negligible.  Consequently,
\[
\sum_{\substack{q\ne q'\in A(X)\\\gcd(d,qq')=1}}
\frac{g(q)g(q')}{qq'}=(1+o(1))(\log_3 X)^2
\]
uniformly in $d$.  Substitution in the preceding second-moment sum yields
\begin{align}
\sum_{n\in[X,2X]}r(n)^2
&=(C+o(1))\frac{X(\log_3 X)^2}{(\log X)^2}
\sum_{\substack{a\ne a'\in B(X)\\2\mid(a-a')}}g(|a-a'|) \notag\\
&\qquad+o(X\log_3 X).
\label{eq:stock}
\end{align}

Next we simplify the expression on the right-hand side of
Equation~\eqref{eq:stock}.
For integers $a\ne a'$, let
\[
 \mathfrak S(\{a,a'\}):=
 \begin{cases}
   C\,g(|a-a'|) & \text{if } 2\mid a-a',\\
   0 & \text{otherwise},
 \end{cases}
\]
denote the singular series of the pair $\{a,a'\}$.  A classical
theorem of Goldston on the mean value of the singular series (see
Montgomery and Soundararajan~\cite[Equation~(48)]{MS04}; a weaker
form, which suffices here, goes back to Bombieri and
Davenport~\cite{BD66}) asserts that
\[
 \sum_{\substack{d_1,d_2\le h\\ d_1\ne d_2}}
 \mathfrak S(\{d_1,d_2\})
 =h^2-h\log h+O(h).
\]
Since $\mathfrak S$ depends only on the difference of its arguments
and $B(X)$ is an interval containing $h:=\log
X/\sqrt{\log_3X}+O(1)$ integers, it follows that
\begin{align*}
\sum_{\substack{a\ne a'\in B(X)\\
  2\mid(a-a')}}g(|a-a'|)
 =\frac1C\sum_{a\ne a'\in B(X)}\mathfrak S(\{a,a'\})
 =\frac{1+o(1)}{C} \frac{(\log X)^2}{\log_3{X}} \, .
\end{align*}%
Substituting the above equation into Equation~\eqref{eq:stock}, we
deduce Equation~\eqref{eq:SecondMoment}.  Together with
Lemma~\ref{lem:far}, this proves, subject to
Hypothesis UH, that
\[
 |\mathcal S(X)|=(1+o(1))X.
\]

For the unconditional lower bound, one instead uses the
uniform upper bound~\eqref{eq:BRUN}; the derivation of
Equation~\eqref{eq:SecondMoment} then shows, \emph{mutatis mutandis},
that
\begin{gather}
 \sum_{n\in[X,2X]}r(n)^2
 \leq\kappa(1+o(1))X\log_3X,
 \qquad \kappa=8.
 \label{eq:END2}
\end{gather}
The Cauchy--Schwarz inequality now yields
$M_0(X)\geq((1/8)+o(1))X$.  After removing the
$o(X)$ integers counted in Lemma~\ref{lem:far}, we obtain
\[
 |\mathcal S(X)|\geq((1/8)+o(1))X.
\]

It remains to pass from these local estimates to the union
$\mathcal S$.  Given $N$, put $X_j:=\lfloor N/2^j\rfloor$ for those
$j\geq1$ with $X_j\geq2^{10}$.  The intervals $[X_j,2X_j]$ have
disjoint interiors; consecutive intervals have a gap of at most one
integer, and
\[
 \sum_jX_j=N+O(\log N).
\]
Since $\mathcal S(X_j)\subseteq\mathcal S$, summing the local bounds
and subtracting $O(\log N)$ for common endpoints gives the following.
For every $\varepsilon>0$, after discarding the finitely many indices
with $X_j$ below the threshold required by the local estimate,
\[
 \#(\mathcal S\cap[1,N])
 \geq(1/8-\varepsilon)N-o(N)
\]
unconditionally.  Hence the lower density of $\mathcal S$ is at least
$1/8$.  Under Martin's uniform Bateman--Horn hypothesis, the same
argument with $1-\varepsilon$ in place of $1/8-\varepsilon$ shows that
the lower density is $1$; therefore $\mathcal S$ has density $1$.

 \appendix

\section{Deferred Proofs}
\label{sec:deferred}
\subsection{Proof of Claim~\ref{claim:vanishing}}
\label{sec:vanishing}
\VANISHING*
\begin{proof}
  Suppose that $R_i(a)\alpha_i^a\sigma(\alpha_i)^q=0$.
The constant term of the characteristic polynomial $\Psi$ is
$-a_{k-1}\ne0$, so every characteristic root is non-zero. Since $a>0$
and every coefficient occurring in $R_i$ is
non-zero, a sub-polynomial consisting of a single monomial cannot
vanish at $a$. Thus $R_i$ has at least two monomials.
Write
 \begin{gather*}
R_i(x) =  b_{i_0}x^{i_0}+b_{i_1} x^{i_1}+\cdots+b_{i_t} x^{i_t},
 \end{gather*}
 where $t\ge 1$, $0\le i_0<i_1<\cdots<i_t\le \nu_i-1$ and $b_{i_0},\ldots,b_{i_t}$ are nonzero algebraic integers. Simplifying across by $x^{i_0}$, we may assume that $x=a$ is a root of 
 \begin{gather*}
 b_{i_0}+b_{i_1} x^{i_1-i_0}+\cdots+b_{i_t}x^{i_t-i_0}.
 \end{gather*}
 But then $a$ divides the norm of $b_{i_0}$, a nonzero integer of
 size  at most $(k^2A)^{(2k^2+1) k!}$. Since $a\in B(X)$, this implies that  
\begin{gather*}
 \frac{\log X}{\sqrt{\log_3 X}}<(kA)^{4(k+2)!},
 \end{gather*}
and  so 
\begin{gather*}
 \log X<2(kA)^{4(k+2)!}
 \log((kA)^{2(k+2)!})<(kA)^{8(k+2)!}<\exp(8(k+2)^{k+2}\log(kA)) \, .
 \end{gather*}
 This implies that
\begin{gather*}
 X<\exp_3((k+2)\log (k+2)+\log(8\log(kA)))
\end{gather*}
Since
$(k+2)\log(k+2)+\log(8\log(kA))\le(20\,k!\log(kA))^2$, it follows from
Lemma~\ref{lem:absorb} that $2X\le\mathcal T$, as required.
\end{proof}

 \subsection{Proof of Claim~\ref{claim:heightbound}}
 \label{sec:technical}

 \TECHNICAL*
\begin{proof}
  Suppose that $2X > \mathcal T$.
We must show that $h(\boldsymbol y)\leq (1+h(\boldsymbol x))\varepsilon$.
  We make a case distinction on the value of $h(a)$.
  We first suppose, for a contradiction, that $h(a)<3k^2\log(k^2 A)$. 
Since $a \geq \frac{\log
  X}{\sqrt{\log_3 X}}$, we have
\begin{gather*}
  \log X <  a \sqrt{\log_3 X} \leq
  \exp(3k^2\log(k^2 A))  {\sqrt{\log_3 X}}.
\end{gather*}
Put $Y:=\log X$ and $U:=\exp(3k^2\log(k^2A))$.  The function
$y/\sqrt{\log\log y}$ is increasing for $y\geq e^e$.  Moreover,
$U\geq40^{12}$ and
\[
 \frac{2U\log U}{\sqrt{\log\log(2U\log U)}}>U.
\]
Hence the preceding inequality $Y<U\sqrt{\log\log Y}$ implies
$Y<2U\log U$.  Therefore
\begin{gather*}
\log X<6k^2\log(k^2 A)\exp (3k^2(\log k^2 A))< \exp(6k^2\log (3k^2 A)),
\end{gather*}
and so 
\begin{gather*}
X<\exp_2(6k^2\log(3k^2 A)). 
\end{gather*}

Since
$\exp_2(6k^2\log(3k^2A))=\exp_3\bigl(\log(6k^2\log(3k^2A))\bigr)$ and
$\log(6k^2\log(3k^2A))\le(20\,k!\log(kA))^2$, Lemma~\ref{lem:absorb}
gives $2X\le\mathcal T$, contradicting the hypothesis of the claim.
So the assumption $h(a)<3k^2\log(k^2 A)$
leads to a contradiction.

Next, we assume that $h(a) \geq 3k^2\log (k^2A)$.  Using the height
bound on the coefficients of the $R_i$ given in~\eqref{eq:height_coef}
and Item~6 of Proposition~\ref{thm:height_bounds} we have that for all
$i \in \{1,\ldots,\ell\}$,
\begin{gather}
    h(R_i(a))  \leq  k ((2k^2+1)\log(k^2A)+\log 2+ h(a)) \leq  2kh(a) \, .
  \label{eq:XX}
\end{gather}
Recall that
$\boldsymbol y=(R_2(a)/R_1(a),\ldots,R_\ell(a)/R_1(a))$.
Using the standard coordinatewise estimate
$h(z_1,\ldots,z_m)\leq \sum_{j=1}^m h(z_j)$, together with
$h(\alpha/\beta)\leq h(\alpha)+h(\beta)$, Equation~\eqref{eq:XX},
and $\ell\leq k$, we obtain
\begin{align*}
h({\boldsymbol y})
&\leq \sum_{i=2}^{\ell} h\left(R_i(a)/R_1(a)\right) \\
&\leq \sum_{i=2}^{\ell}\left(h(R_i(a))+h(R_1(a))\right) \\
&\leq 4k(\ell-1)h(a)
 \leq 4k^2h(a)
 =4k^2\log a
 <4k^2\log_2 X \, .
\end{align*}

Next we give a lower bound on $h({\boldsymbol x})$.  First observe that,
for any two distinct characteristic roots $\alpha_i$ and $\alpha_j$, the
quotient $\alpha_i/\alpha_j$ is not a root of unity by the
non-degeneracy of $\boldsymbol u$, and it has degree at most $k^2$.
If this quotient has degree at least two, Item~2 of
Proposition~\ref{thm:height_bounds} therefore applies.  If it has degree one,
then it is a non-zero rational number different from $\pm1$, and hence its
height is at least $\log 2$.  Since $k\geq2$ and
$\log(3k^2)\geq\log 12>2$, we have
\[
\frac{2}{k^2\bigl(\log(3k^2)\bigr)^3}\leq\frac1{16}<\log2.
\]
Consequently, in all cases,
\begin{gather}
h(\alpha_i/\alpha_j)\geq
 \frac{2}{k^2\bigl(\log(3k^2)\bigr)^3}
 \qquad (i\neq j).
\label{eq:dob}
\end{gather}

Recall that
$\boldsymbol{x}=\boldsymbol{\gamma}^{a+\eta q}$, where
$\boldsymbol{\gamma}=(\alpha_2/\alpha_1,\ldots,
\alpha_\ell/\alpha_1)$, and put $\beta:=\alpha_2/\alpha_1$.
For $X>55$ we have
\[
q\leq\sqrt{\log X}<\frac{\log X}{2\sqrt{\log_3 X}}
\qquad\text{and}\qquad
a\geq\frac{\log X}{\sqrt{\log_3 X}},
\]
and hence
$a+\eta q\geq a-q>\log X/(2\sqrt{\log_3 X})>0$.
The standard coordinate lower bound for affine height, together with
Item~3 of Proposition~\ref{thm:height_bounds}, now gives
\begin{align*}
h({\boldsymbol x})
&\geq h\bigl(\beta^{a+\eta q}\bigr)
 =(a+\eta q)h(\beta)
 \geq(a-q)h(\beta)\\
&\geq \left(\frac{\log X}{2\sqrt{\log_3 X}}\right)h(\beta)
 \geq \sqrt{\log X}\,h(\beta)
 \qquad {\text{\rm for}}\quad X>55,
\end{align*}
since $\log X/(2\sqrt{\log_3 X})>\sqrt{\log X}$ for $X>55$.
Combining this with~\eqref{eq:dob}, we obtain
\[
h({\boldsymbol x})\geq
 \frac{2\sqrt{\log X}}
 {k^2\bigl(\log(3k^2)\bigr)^3}.
\]
Suppose for a contradiction that $h({\boldsymbol y}) > (1+h({\boldsymbol
  x}))\varepsilon$.  Then, combining the upper bound on
$h(\boldsymbol y)$ and the lower bound on $h(\boldsymbol x)$, we get
\begin{gather*}
2k^4\bigl(\log(3k^2)\bigr)^3 (8k)^{6k^3}\log_2 X
> \sqrt{\log X} \, .
\end{gather*}
Putting
\[
c:=2k^4\bigl(\log(3k^2)\bigr)^3(8k)^{6k^3},
\]
we may apply Corollary~\ref{cor:elementary}: for $k\ge2$,
using $\log k\leq k$, $\log\log(3k^2)\leq k$, and $\log(8k)\leq2k$,
\begin{gather*}
\log c \, =\, \log2+4\log k+3\log\log(3k^2)+6k^3\log(8k) \, <\, 20k^4,
\end{gather*}
so the Corollary yields $2X<\exp_5(10^{10}k^6)\le\mathcal T$, contrary
to the hypothesis of the claim.  We conclude that
$h(\boldsymbol{y}) \leq (1+h(\boldsymbol{x}))\varepsilon$.
\end{proof}

\subsection{Proof of Claim~\ref{claim:easy}}
\label{sec:easy}
\EASY*
\begin{proof}
Suppose that \eqref{eq:11} fails at two distinct
exponent values $a+\eta q\ne a'+\eta q'$.
For the vectors ${\boldsymbol x},{\boldsymbol y}$ and
${\boldsymbol x}',{\boldsymbol y}'$ respectively corresponding to
$(q,a)$ and $(q',a')$, we have
$h({\boldsymbol y})>(1+h({\boldsymbol x}))\varepsilon$ and
$h({\boldsymbol y}')>(1+h({\boldsymbol x}'))\varepsilon$.
Under the hypothesis
$2X>\mathcal T$, the first case in the
proof of Claim~\ref{claim:heightbound} shows that both $h(a)$ and
$h(a')$ are at least $3k^2\log(k^2A)$, so Equation~\eqref{eq:XX}
applies to both $a$ and $a'$; moreover $h(a),h(a')\leq\log_2X$ for
sufficiently large $X$.
Applying the coordinatewise estimate used in the proof of
Claim~\ref{claim:heightbound}, with $R_2$ in place of $R_1$, gives
$h({\boldsymbol y})\leq 4k^2\log a$ and
$h({\boldsymbol y}')\leq 4k^2\log a'$.  Adding the two failed height
inequalities, we therefore get
\begin{align}
\label{eq:morethan1}
8k^2\log_2 X
& >  4k^2\log a+4k^2\log a'
\geq h({\boldsymbol y})+h({\boldsymbol y}')
> (2+h({\boldsymbol x})+h({\boldsymbol x'}))\varepsilon\nonumber\\
& \geq  (2+h({\boldsymbol x}/{\boldsymbol x}'))\varepsilon.
\end{align}
For the right--most inequality above see equation (7.6) in \cite{ESS}.  But in ${\boldsymbol x}/{\boldsymbol x'}$, the unknown vector ${\bf A}$ is gone and 
\begin{gather*}
h({\boldsymbol x}/{\boldsymbol x}')=h((\alpha_3/\alpha_2)^{(a+\eta q)-(a'+\eta q')},\ldots,(\alpha_{\ell'}/\alpha_2)^{(a+\eta q)-(a'+\eta q')}). 
\end{gather*}
In particular, by~\eqref{eq:dob} and
the absence of correlated representations,
$|(a+\eta q)-(a'+\eta q')|\geq\sqrt{\log X}$, we have
\begin{align*}
h({\boldsymbol x}/{\boldsymbol x'}) & \ge  |(a+\eta q)-(a'+\eta q')| \min \{h(\alpha_i/\alpha_j): i\ne j\in \{2,3,\ldots,\ell'\}\}\\
& \ge  \frac{2\sqrt{\log X}}{k^2(\log(3k^2))^3} \,  .
\end{align*}
So the estimate~\eqref{eq:morethan1} leads to
\begin{gather*}
8k^2\varepsilon^{-1}\log_2 X
> \frac{2\sqrt{\log X}}{k^2\bigl(\log(3k^2)\bigr)^3} \, ,
\end{gather*}
and hence $\sqrt{\log X}<c\log_2 X$
 for
\[
c:=4k^4\bigl(\log(3k^2)\bigr)^3(8k)^{6k^3}.
\]
As in the proof of Claim~\ref{claim:heightbound} we have
$\log c<20k^4$ for $k\ge2$, so Corollary~\ref{cor:elementary} yields
$2X<\exp_5(10^{10}k^6)\le\mathcal T$, which contradicts the hypothesis
$2X>\mathcal T$.

\end{proof}

\section*{Acknowledgements}

{\sloppy
Florian Luca is also affiliated with the Max Planck Institute for Software Systems,
Saarbr\"ucken, Germany. Florian Luca and Jo\"el Ouaknine are supported
by ERC grant DynAMiCs (101167561).
James Maynard is supported by ERC grant PraDa (851318).
Armand Noubissie is supported by FWF grant DOI 10.55776/ESP1571325.
Jo\"el Ouaknine is also affiliated with Keble College, Oxford as an
\texttt{emmy.network} Fellow, and supported by DFG grant 389792660 as part of
TRR 248. James Worrell is supported by UKRI Fellowship EP/X033813/1.\par}

\bibliographystyle{alpha}

\bibliography{universal}

\end{document}